\documentclass[submission,copyright,creativecommons]{eptcs}
\providecommand{\event}{SR 2013} 
\usepackage{breakurl}             

\usepackage{amsmath}
\usepackage{stmaryrd}
\usepackage[ruled, vlined]{algorithm2e}

\usepackage{amsthm}
\newtheorem{theorem}{Theorem}

\newtheorem{remark}{Remark}

\usepackage{tikz}
\usetikzlibrary{arrows}

\usepackage{tabularx}

\providecommand{\dontprintsemicolon}{\DontPrintSemicolon}
\providecommand{\linesnumberedhidden}{\LinesNumberedHidden}
\providecommand{\showln}{\ShowLn}

\usepackage{microtype}

\title{Reasoning about Strategies under Partial Observability\\and Fairness Constraints}
\author{
	Simon Busard, Charles Pecheur\thanks{This work is supported by the European Fund for Regional Development and by the Walloon Region.}
	\institute{
		ICTEAM Institute,\\
		Universit\'e catholique de Louvain,\\
		Louvain-la-Neuve, Belgium
	}
	\email{simon.busard@uclouvain.be}
	\email{charles.pecheur@uclouvain.be}
	\and
	Hongyang Qu
	\institute{
		Dept. of Computer Science,\\
		University of Oxford,\\
		Oxford, United Kingdom
	}
	\email{Hongyang.Qu@cs.ox.ac.uk}
	\and
	Franco Raimondi
	\institute{
                Dept. of Computer Science,\\
		Middlesex University,\\
		London, United Kingdom
	}
	\email{f.raimondi@mdx.ac.uk}
}
\def\titlerunning{Reasoning about Strategies under Partial Observability and Fairness Constraints}
\def\authorrunning{S. Busard, C. Pecheur, H. Qu, F. Raimondi}
\begin{document}
\maketitle

\vspace{-1em}
\begin{abstract}
A number of extensions exist for Alternating-time Temporal Logic; some of these mix strategies and partial observability but, to the best of our knowledge, no work provides a unified framework for strategies, partial observability and fairness constraints. In this paper we propose $ATLK^F_{po}$, a logic mixing strategies under partial observability and epistemic properties of agents in a system with fairness constraints on states, and we provide a model checking algorithm for it.
\end{abstract}
\vspace{-1em}

\section{Introduction}

A number of extensions exist for Alternating-time Temporal Logic; starting from~\cite{Hoek-Wooldridge-03}, partial observability has been investigated by many authors, see for instance~\cite{jamroga2007constructive} and references therein. But, to the best of
our knowledge, no work provides a unified framework for strategies,
partial observability and fairness constraints.
For example, Jamroga and van der Hoek
proposed, among other logics, ATOL, mixing partial observability with
strategies of agents~\cite{Jamroga-Hoek-04}. Along the same lines,
Schobbens studied ATL$_{ir}$\cite{Schobbens-04}, seen as the minimal
ATL-based logic for strategies under partial
observability~\cite{Jamroga-Dix-06}.
On the other hand, some efforts have been made on bringing fairness to
ATL. For instance the work of Alur et
al.~\cite{Alur-Henzinger-others-02}, or the work of Kl\"uppelholz and
Baier~\cite{Kluppelholz-Baier-08} introduce the notion of fairness
constraints on actions, asking for an infinitely often enabled action
to be taken infinitely often.
For temporal and epistemic logics, however, fairness
conditions are normally provided on \emph{states}.
Furthermore, it has been shown that (weak, strong or unconditional) fairness constraints on actions, can be reduced to (weak, strong or unconditional, respectively) fairness constraints on states (see \cite{Baier-Katoen-08}, for instance).
In this paper we
propose $ATLK^F_{po}$, a logic mixing strategies under partial
observability and epistemic properties of agents in a system with
unconditional fairness constraints \emph{on states}, and we provide a model checking
algorithm for it.

To motivate the need for fairness constraints in ATL under partial
observability, consider the simple card game example in
\cite{Jamroga-Hoek-04}. The game is played between a player and a
dealer. It uses three cards, $A$, $K$ and $Q$; $A$ wins over $K$, $K$
wins over $Q$ and $Q$ wins over $A$. First, the dealer gives one card
to the player, keeps one and leaves the last one on table. Then the
player can keep his card or swap it with the one on the table. The
player wins if his card wins over the dealer's card.  Under ATL$_{ir}$
semantics, the player cannot win the game: he cannot distinguish
between, for example, $<A,K>$ and $<A,Q>$ (where $<a,b>$ means "player
has card $a$, dealer has card $b$") and thus has to make the same
action in both states, with a different result in each case.  Consider
now a variation of this game: the game does not terminate after the
first round. Instead, if the player does not win, cards are
redistributed. In this case, too, the player cannot win the game: for
instance, he will have to choose between keeping or swapping cards in
$<A,K>$ and $<A,Q>$, so he won't be able to enforce a win because the
dealer (that chooses the given cards) can be unfair and always give
the losing pair.  But if we add one fairness constraint per
intermediate state---i.e. the states in which the player has to choose
between swapping or keeping---the player has a strategy to finally win
the game. In this case, we only consider paths along which all
fairness constraints are met infinitely often: this situation
corresponds to a fair dealer, giving the cards randomly. The player
can thus finally win because $<A,K>$ will eventually happen---even if
he cannot distinguish it from $<A,Q>$---, so he knows a strategy to
win at least a round: keeping his card.

Another example of application of fairness constraints in ATL is Multi-Agent Programs~\cite{Dastani-Jamroga-10}. These programs are composed of interleaved agent programs and fairness constraints are used to avoid unfair interleaving.
Dastani and Jamroga express fairness as formulae of the logic ATL*~\cite{Dastani-Jamroga-10}; in this paper, instead, we deal only with ATL and therefore fairness constraints cannot be expressed as formulae of the logic. The situation is similar to the case of LTL versus CTL model checking: in the first case model checking fairness is reduced to model checking a more complex formula using the same verification algorithms; in the second case fairness is incorporated into bespoke verification algorithms. In our work we chose ATL over ATL* because of complexity considerations (see Section~\ref{section:mc}).

The rest of the paper is structured as follows: Section~\ref{section:logics} presents the syntax and semantics of $ATLK^F_{po}$ and Section~\ref{section:mc} presents two model checking algorithms for the logic. Finally, Section~\ref{section:conclusion} summarizes the contribution and draws some future work.

\section{Syntax and Semantics}\label{section:logics}

This section presents the syntax and semantics of $ATLK^F_{po}$, an
extension of ATL with partial observability under fairness constraints
on states. An extension with full observability under the same fairness constraints $ATLK^F_{fo}$ is also presented because the model
checking algorithm for $ATLK^F_{po}$ relies on the one for
$ATLK^F_{fo}$.

\paragraph{Syntax} Both logics share the same syntax, composed of
the standard Boolean connectors ($\vee$, $\wedge$, $\neg$, etc.), CTL
operators ($EX$, $EU$, $EG$, etc.)~\cite{Clarke-Grumberg-others-99},
knowledge operators ($K_{ag}$, $E_\Gamma$, $D_\Gamma$,
$C_\Gamma$)~\cite{Fagin-Halpern-others-95} and strategic operators
($\langle\Gamma\rangle X$, $\langle\Gamma\rangle G$,
$\langle\Gamma\rangle U$, $\langle\Gamma\rangle W$ and their $[\Gamma]$
counterparts)~\cite{Alur-Henzinger-others-02}.

\paragraph{Models and notation} $ATLK^F_{fo}$ and $ATLK^F_{po}$
formulae are interpreted over models $M = \langle Ag, S, Act, T, I,$
$\{\sim_i\}, V, F \rangle$ where (1) $Ag$ is a set of $n$ agents; (2)
$S = S_1 \times ... \times S_n$ is a set of global states, each of
which is composed of $n$ local states, one for each agent; (3) $Act =
Act_1 \times ... \times Act_n$ is a set of joint actions, each of
which is composed of $n$ actions, one for each agent; (4) $T \subseteq
S \times Act \times S$ is a transition relation between states in $S$
and labelled with joint actions (we write $s\xrightarrow{a} s'$ if
$(s, a, s') \in T$); (5) $I \subseteq S$ is the a set of initial
states; (6) $\{\sim_i\}$ is a set of equivalence relations between
states, and $\sim_i$ partitions the set of states in terms of
knowledge of agent $i$---$s \sim_i s'$ iff $s_i = s'_i$, i.e two states are indistinguishable for agent $i$ if they share the same local state for $i$; (7) $V : S \rightarrow 2^{AP}$ labels states
with atomic propositions of $AP$; (8) $F \subseteq 2^S$ is a set of
fairness constraints, each of which is a subset of states.

A joint action $a = (a_1, ..., a_n)$ \textit{completes} a partially
joint action $a_\Gamma = (a'_i, ..., a'_j)$ composed of actions of
agents in $\Gamma\subseteq Ag$---written $a_\Gamma \sqsubseteq a$---if
actions in $a$ for agents in $\Gamma$ correspond to actions in
$a_\Gamma$.  Furthermore, we define the function $img : S \times Act
\rightarrow 2^S$ as $img(s, a) = \{s' \in S | s\xrightarrow{a}s'\}$,
i.e. $img(s,a)$ is the set of states reachable in one step from $s$
through $a$.

A model $M$ represents a non-deterministic system where each agent has
an imperfect information about the current global state. One
restriction is made on $T$: $\forall s, s' \in S, s \sim_i s' \implies
enabled(s, i) = enabled(s', i)$ where $enabled(s, i) = \{a_i \in
Act_i | \exists s' \in S, a \in Act \textrm{ s.t. } (a_i)
\sqsubseteq a \wedge s \xrightarrow{a} s'\}$.  This means that the
actions an agent can perform in two epistemically equivalent states
are the same. The $enabled$ function is straightforwardly extended to
groups of agents.

A \textit{path} in a model $M$ is a sequence $\pi = s_0
\xrightarrow{a_1} s_1 \xrightarrow{a_2}...$ of elements of $T$. We use
$\pi(d)$ for $s_d$. A state $s$ is \textit{reachable} in $M$ if there
exist a path $\pi$ and $d \geq 0$ such that $\pi(0) \in I$ and $\pi(d)
= s$.  A path $\pi$ is \textit{fair} according to a set of fairness
conditions $F = \{f_1, ..., f_k\}$ if for each fairness condition
$f$, there exist infinitely many positions $d \geq 0$ such that
$\pi(d) \in f$. A state $s$ is \textit{fair} if there exists a fair
path starting at $s$.

A \textit{strategy} for agent $i$ is a function $f_i:S \rightarrow
Act_i$ where, for any state $s$, $f_i(s) \in enabled(s, i)$; a
strategy maps each state to an enabled action. We call these
strategies \textit{global strategies}. A \textit{uniform strategy} for
agent $i$ is a global strategy $f_i$ where $\forall s, s' \in S, s'
\sim_i s \implies f_i(s) = f_i(s')$, i.e. agent $i$ cannot choose two
different actions for two indistinguishable states.  The
\textit{strategy outcomes} from a state $s$ for a strategy $f_i$,
denoted with $out(s, f_i)$, is the set of paths a strategy can
enforce, i.e. $out(s, f_i) = \{\pi = s_0 \xrightarrow{a_1} s_1 ... |
s_0 = s \wedge \forall d \geq 0, s_{d+1} \in img(s_d, a_{d+1}) \wedge
(f_i(s_d)) \sqsubseteq a_{d+1}\}$. The definition of outcomes is
naturally extended to sets of strategies for a subset of agents.

\paragraph{Semantics} The semantics of both logics are defined
over states of a model $M$ by defining the relations $M, s
\models^F_{fo} \phi$ and $M, s \models^F_{po} \phi$, for $ATLK^F_{fo}$
and $ATLK^F_{po}$, respectively. $M$ can be omitted when clear from the
context. Both relations share a part of their semantics; we write $s \models^F \phi$ if $s \models^F_{fo} \phi$ and $s \models^F_{po} \phi$.  The $s \models^F_{fo} \phi$ and $s
\models^F_{po} \phi$ relations are recursively defined over the
structure of $\phi$ and follow the standard interpretation for most of
the operators. $s\models^F p$ if $p \in V(s)$; $\vee$ and $\neg$ are
interpreted in the natural way.  $s \models^F K_i \phi$ if $\phi$
is true in all fair reachable states indistinguishable from $s$ for agent $i$, $s
\models^F E_\Gamma \phi$ if all agents in $\Gamma$ know $\phi$,
$s\models^F D_\Gamma \phi$ if, by putting all their knowledge in
common, agents of $\Gamma$ would know $\phi$, and $s \models^F
C_\Gamma \phi$ if $\phi$ is common knowledge among agents of
$\Gamma$~\cite{Fagin-Halpern-others-95}.  $s \models^F E \psi$ if
there is a path $\pi$ starting at $s$ satisfying $\psi$, $\pi
\models^F X \phi$ if $\pi(1)$ satisfies $\phi$, $\pi \models^F \phi_1
U \phi_2$ if $\phi_1$ is true along the path until $\phi_2$ is true, $\pi \models G \phi$ if $\phi$ is always true along
$\pi$, and $\pi \models \phi_1 W \phi_2$ if $\pi \models (\phi_1 U \phi_2) \vee G \phi_1$~\cite{Clarke-Grumberg-others-99}.

The meaning of the $\langle\Gamma\rangle$ operator is different in the
two semantics:

\noindent (i) $s \models^F_{fo} \langle \Gamma \rangle \psi$ iff there exists a set of \textbf{global strategies}
$f_\Gamma$, one for each agent in $\Gamma$, such that for all
\textbf{fair paths} $\pi \in out(s, f_\Gamma), \pi \models^F \psi$;

\noindent (ii) $s \models^F_{po} \langle \Gamma \rangle \psi$ iff there exists a set of \textbf{uniform strategies}
$f_\Gamma$, one for each agent in $\Gamma$, such that for all $s' \sim_\Gamma s$, for all
\textbf{fair paths} $\pi \in out(s', f_\Gamma), \pi \models^F
\psi$.
    
The $[\Gamma]$ operator is the dual of $\langle\Gamma\rangle$: $s\models^F [\Gamma] \psi$ iff $s\models^F \neg \langle\Gamma\rangle \neg \psi$.

\section{Model Checking $ATLK^F_{fo}$ and $ATLK^F_{po}$}\label{section:mc}

\paragraph{Model checking $ATLK^F_{fo}$}
The model checking algorithm for $ATLK^F_{fo}$ is defined by the function $\llbracket
. \rrbracket^F_{fo} : ATLK^F_{fo} \rightarrow 2^S$ returning the set
of states of a given model $M$ satisfying a given $ATLK^F_{fo}$
property. This function is defined in the standard way for Boolean
connectors, CTL and knowledge
operators~\cite{Clarke-Grumberg-others-99,Lomuscio-Penczek-07}.
The $[\Gamma]$ operators are evaluated as follows:
\begin{align*}
	\llbracket [\Gamma] X \phi \rrbracket^F_{fo} &= Pre_{[\Gamma]}(\llbracket \phi \rrbracket^F_{fo} \cap Fair_{[\Gamma]}) \\
	\llbracket [\Gamma] \phi_1 U \phi_2 \rrbracket^F_{fo} &= \mu Z . (\llbracket \phi_2 \rrbracket^F_{fo} \cap Fair_{[\Gamma]}) \cup (\llbracket \phi_1 \rrbracket^F_{fo} \cap Pre_{[\Gamma]}(Z)) \\
	\llbracket [\Gamma] G \phi \rrbracket^F_{fo} &= \nu Z . \llbracket \phi \rrbracket^F_{fo} \cap \bigcap_{f \in F} Pre_{[\Gamma]}(\mu Y . (Z \cap f) \cup (\llbracket \phi \rrbracket^F_{fo} \cap Pre_{[\Gamma]}(Y))) \\
	\llbracket [\Gamma] \phi_1 W \phi_2 \rrbracket^F_{fo} &= \begin{array}{l}\nu Z . (\llbracket\phi_2\rrbracket^F_{fo} \cap Fair_{[\Gamma]}) \\\phantom{\nu Z . } \cup (\llbracket\phi_1\rrbracket^F_{fo} \cap \bigcap_{f \in F} Pre_{[\Gamma]}(\mu Y . (\llbracket\phi_2\rrbracket^F_{fo}\cap Fair_{[\Gamma]}) \cup (Z \cap f) \cup (\llbracket \phi_1 \rrbracket^F_{fo} \cap Pre_{[\Gamma]}(Y))))\end{array}
\end{align*}
where $Pre_{[\Gamma]}(Z)$ = $\{s | \forall a_\Gamma \in enabled(s,
\Gamma), \exists a \textrm{ s.t. } a_\Gamma \sqsubseteq a \wedge
img(s, a) \cap Z \neq \emptyset \}$ and $Fair_{[\Gamma]}$ =
$\llbracket [\Gamma] G~true \rrbracket^F_{fo}$. $\mu Z.\tau(Z)$ and $\nu Z.\tau(Z)$ are the least and greatest fix points of function $\tau(Z)$.
Intuitively, the $Pre_{[\Gamma]}(Z)$ operator returns the set of
states in which $\Gamma$ cannot avoid to reach a state of $Z$. Thus,
$\llbracket [\Gamma] G \phi \rrbracket^F_{fo}$ returns the set of
states in which $\Gamma$ cannot avoid a path of states of $\llbracket
\phi \rrbracket^F_{fo}$ going through all fairness constraints
infinitely often; $Fair_{[\Gamma]}$ is the set of states in which
$\Gamma$ cannot avoid a fair path.
Note that the $\langle \Gamma\rangle$ operators can be computed using
the $[\Gamma]$ and $\neg$ operators, but can also be computed directly
using the dual forms from the ones above. For example $\llbracket
\langle\Gamma\rangle G \phi \rrbracket^F_{fo}$ = $\nu Z . (\llbracket
\phi \rrbracket^F_{fo} \cup \overline{Fair_{[\Gamma]}}) \cap
Pre_{\langle\Gamma\rangle}(Z)$, where $Pre_{\langle\Gamma\rangle}(Z) = \overline{Pre_{[\Gamma]}(\overline{Z})} =
\{s | \exists a_\Gamma \in enabled(s, \Gamma) \textrm{ such that }
\forall a, a_\Gamma \sqsubseteq a \implies img(s, a) \subseteq Z\}$. $\overline{Z} \subseteq S$ is the complement of the set $Z \subseteq S$.

The correctness of the model checking algorithm for $ATLK^F_{fo}$ 
follows from Theorem~\ref{thm:correct-mc-weak}.
\begin{theorem}\label{thm:correct-mc-weak}
	For all states $s \in S$, $s \models^F_{fo} \phi$ if and only if $s \in \llbracket \phi \rrbracket^F_{fo}$.
\end{theorem}
\begin{proof}[Proof sketch]
First, $Reach_{[\Gamma]}(P_1, P_2) = \mu Y . P_2 \cup (P_1
\cap Pre_{[\Gamma]}(Y))$ computes the set of states in which
$\Gamma$ cannot avoid a finite path of states of $P_1$ to a state of
$P_2$. We can prove it by induction over the computation of the
least fix point. It is true by definition of the least fix point and
the $Pre_{[\Gamma]}$ operation.

Then, for the $[\Gamma] G \phi$ operator, $\llbracket [\Gamma] G \phi
\rrbracket^F_{fo}$ $=$ $\nu Z . \llbracket \phi \rrbracket^F_{fo} \cap
\bigcap_{f \in F} Pre_{[\Gamma]}(\mu Y . (Z \cap f) \cup
(\llbracket \phi \rrbracket^F_{fo} \cap Pre_{[\Gamma]}(Y)))$ $=$ $\nu Z . \llbracket \phi \rrbracket^F_{fo} \cap
\bigcap_{f \in F} Pre_{[\Gamma]}(Reach_{[\Gamma]}(\llbracket \phi \rrbracket^F_{fo}, Z \cap f))$
computes the set of states in which $\Gamma$ cannot avoid a fair
path (i.e. going through each $f \in F$ infinitely often)
that satisfies $G \phi$. We prove it by induction over the
computation of the greatest fix point and by using what has been
proved just above.
	
Thanks to this, we can easily prove that $Fair_{[\Gamma]} =
\llbracket [\Gamma] G true \rrbracket^F_{fo}$ computes the set of
states in which $\Gamma$ cannot avoid a fair path (it is just a
particular case of the $[\Gamma] G$ operator).
	
Then, $[\Gamma]X$ and $[\Gamma] U$
operators compute the set of states in which
$\Gamma$ cannot avoid a successor in
$\llbracket\phi\rrbracket^F_{fo}$ in which $\Gamma$ cannot
avoid a fair path, respectively in which $\Gamma$ cannot avoid a
finite path through states of $\llbracket\phi_1\rrbracket^F_{fo}$ to
a state of $\llbracket\phi_2\rrbracket^F_{fo}$, in which $\Gamma$
cannot avoid a fair path. In particular, the proof for $[\Gamma] U$ directly follows from the proof for $Reach_{[\Gamma]}$.
	
Finally, the proof for the $[\Gamma]W$ operator is similar to the
one for $[\Gamma]G$ operator.
The proof of correctness of the algorithms for $\langle
\Gamma \rangle$ operators follows from the proof for $[\Gamma]$
operators, the duality of these operators and standard fix point
properties.
\end{proof}

\paragraph{Model checking $ATLK^F_{po}$ -- basic algorithm}
A basic algorithm is presented in Algorithm~\ref{algo:eval<>}.  It
relies on the model checking algorithm for $ATLK^F_{fo}$.  It uses two
sub-algorithms: $Split$ and $\llbracket . \rrbracket^F_{fo}|_{strat}$,
where $strat$ is a strategy represented as a set of state/action
pairs. The latter is a modified version of the algorithm described in
the previous section with $Pre_{\langle\Gamma\rangle}|_{strat}$
replacing $Pre_{\langle\Gamma\rangle}$ where
$Pre_{\langle\Gamma\rangle}|_{strat}(Z) = \{s | \exists a_\Gamma \in
enabled(s, \Gamma) \textrm{ such that } \langle s, a_\Gamma \rangle
\in strat \wedge \forall a, a_\Gamma \sqsubseteq a \implies img(s, a)
\subseteq Z\}$,
i.e., $Pre_{\langle\Gamma\rangle}|_{strat}(Z)$ is
$Pre_{\langle\Gamma\rangle}(Z)$ restricted to states and actions
allowed by $strat$. Furthermore, $\llbracket
. \rrbracket^F_{fo}|_{strat}$ recursively calls $\llbracket
. \rrbracket^F_{po}$ on sub-formulae, instead of $\llbracket
. \rrbracket^F_{fo}$.

\begin{algorithm}[h!t]
	\dontprintsemicolon
	\linesnumberedhidden
	\KwData{$M$ a given (implicit) model, $\Gamma$ a subset of agents of $M$, $\psi$ an $ATLK^F_{po}$ path formula.}
	\KwResult{The set of states of $M$ satisfying $\langle \Gamma \rangle \psi$.}
	\BlankLine
	$sat = \{\}$\;
	\For{$strat \in Split(S \times Act_\Gamma)$}{
		$winning =  \llbracket \langle \Gamma \rangle \psi\rrbracket^F_{fo}|_{strat}$\;
		$sat = sat \cup \{s \in winning | \forall s' \sim_\Gamma s, s' \in winning\}$\;
	}
	\Return $sat$
	\caption{$\llbracket \langle \Gamma \rangle \psi\rrbracket^F_{po}$}
	\label{algo:eval<>}
\end{algorithm}

The $Split$ algorithm is given in Algorithm~\ref{algo:split}. $Split(S
\times Act_\Gamma)$ returns the set of uniform strategies of the system
(a uniform strategy is represented by the action for group $\Gamma$ allowed in each
state, and this action needs to be the same for each state
in the same equivalence class).

\begin{algorithm}[h!t]
	\dontprintsemicolon
	\linesnumberedhidden
	\KwData{$Strats \subseteq S \times Act_\Gamma$.}
	\KwResult{The set of all the largest subsets $SA$ of $Strats \subseteq S \times Act_\Gamma$ such that no conflicts appear in $SA$.}
	\BlankLine
	$C = \{ \langle s, a_\Gamma \rangle \in Strats | \exists \langle s', a'_\Gamma \rangle \in Strats ~s.t.~ s' \sim_\Gamma s \wedge a_\Gamma \neq a'_\Gamma \}$\;
	\lIf{$C = \emptyset$}{
		\Return $\{Strats\}$\;
	}
	\Else{
		$\langle s, a_\Gamma \rangle = $ pick one in $C$\;
		$E = \{ \langle s', a'_\Gamma \rangle \in Strats | s' \sim_\Gamma s \}$\;
		$A = \{ a_\Gamma \in Act_\Gamma | \exists \langle s, a_\Gamma \rangle \in E \}$\;
		$strats = \{\}$\;
		\For{$a_\Gamma \in A$}{
			$S = \{ \langle s', a_\Gamma \rangle \in E  | a'_\Gamma = a_\Gamma \}$\;
			$strats = strats \cup Split(S \cup (Strats \backslash E))$
		}
		\Return $strats$
	}
	\caption{$Split(Strats)$}
	\label{algo:split}
\end{algorithm}

Intuitively, Algorithm~\ref{algo:eval<>} computes, for each possible uniform strategy $strat$, the set of states for which the strategy is winning, and then keeps only the states $s$ for which the strategy is winning for all states equivalent to $s$.

Before proving the correctness of the basic algorithm, let's prove the correctness of the $Split$ algorithm.

\begin{theorem}\label{th:split}
$Split(Strats)$ computes the set of all the largest subsets $SA$ of $Strats \subseteq S\times Act_\Gamma$ such that no conflicts appear in $SA$.
\end{theorem}

\begin{remark}
A conflict appears in $SA\subseteq S\times Act_\Gamma$ if there exist two elements $\langle s, a_\Gamma\rangle$ and $\langle s', a'_\Gamma \rangle$ in $SA$ such that $s' \sim_\Gamma s$ and $a_\Gamma \neq a'_\Gamma$, i.e. there is a conflict if $SA$ proposes two different actions in two equivalent states.
\end{remark}

\begin{proof}[Proof sketch of Theorem~\ref{th:split}]
$Split$ gets all the conflicting elements of $Strats$. If there are no such elements, then $Strats$ is its own largest non-conflicting subset; otherwise, $Split$ takes one conflicting equivalence class $E$ and, for each of its largest non-conflicting subsets $S$---i.e. subsets of states using the same action---it calls $Split$ on the rest of $Strats$ augmented with the non-conflicting subset $S$.

We can prove the correctness of $Split$ by induction over the number of conflicting equivalence classes of $Strats$. If $Strats$ does not contain any conflicting equivalence classes, $Strats$ is its own single largest subset in which no conflicts appear.
Otherwise, let's assume that $Split(Starts \backslash E)$ with $E$ a conflicting equivalence class of $Strats$ returns the set of all the largest non-conflicting subsets of $Strats \backslash E$; then, by what has been explained above, $Split$ returns the cartesian product between all the largest non-conflicting subsets of $E$ and all the largest non-conflicting subsets of $Strats \backslash E$. Because these cannot be conflicting as they belong to different equivalence classes, we can conclude that $Split$ returns the set of the largest non-conflicting subsets of $Strats$.
\end{proof}


The correctness of Algorithm~\ref{algo:eval<>} is then given by the
following theorem.

\begin{theorem}
$\llbracket \langle \Gamma \rangle \psi\rrbracket^F_{po}$ computes the set of states of $M$ satisfying $\langle \Gamma \rangle \psi$, i.e. 
\begin{equation*}
\forall s \in S, s \in \llbracket \langle \Gamma \rangle \psi\rrbracket^F_{po} \textrm{ iff } s \models^F_{po} \langle \Gamma \rangle \psi.
\end{equation*}
\end{theorem}
\begin{proof}[Proof sketch]
  First, $Split(S \times Act_\Gamma)$ returns all the
possible uniform strategies of the system, where a uniform strategy is
represented by the only action allowed in each equivalence class
of states---states equivalent in terms of the knowledge of
$\Gamma$---, this action being the same for every state of the class.

Indeed, the set of the largest non-conflicting subsets of $S \times
Act_\Gamma$ is the set of possible uniform strategies. A non-conflicting
subset of $S \times Act_\Gamma$ provides at most one action for each
equivalence class of states, otherwise it would not be
non-conflicting; second, a largest non-conflicting subset of $S
\times Act_\Gamma$ provides exactly one action for each equivalence
class of states, otherwise there would be a larger subset giving one
action for the missing equivalence classes and this subset would not
be conflicting.
Finally, a largest non-conflicting subset of $S
\times Act_\Gamma$ is a uniform strategy because it is exactly the
definition of a uniform strategy: giving one possible action for each
equivalence class. This thus ends the proof that $Split$ returns the
set of all possible uniform strategies.

  Second, $winning = \llbracket \Gamma \rrbracket \psi\rrbracket^F_{fo}
  \psi|_{strat}$ returns the set of states for which the strategy $strat$ is
  winning. Indeed, it uses $ATLK^F_{fo}$ model checking algorithm, restricted to actions in $strat$. It thus returns the set of states for which there is a (global) winning strategy in $strat$. As $strat$ is, by construction, a uniform strategy, $winning$ is the set of states for which there exists a uniform winning strategy---in fact, it is $strat$ itself.
  
  Finally, the set $\{s \in winning | \forall s' \sim_\Gamma s, s' \in winning\}$ is the set of states $s$ for which $strat$ is a winning strategy for all $s' \sim_\Gamma s$. $sat$ thus accumulates all the states $s$ for which there is a winning strategy for all states indistinguishable from $s$. As this is exactly the semantics of the property, i.e. $sat$ is exactly the set of states of the system satisfying the property, the proof is done.
\end{proof}

\paragraph{Improving the basic algorithm}
The first improvement proposed for the basic algorithm is the
pre-filtering of states to the ones satisfying the property under
$ATLK^F_{fo}$ ; we can filter them because if a state $s$ does not
satisfy $\langle \Gamma\rangle \psi$ under $ATLK^F_{fo}$, $s$ cannot
satisfy $\langle\Gamma\rangle \psi$ under $ATLK^F_{po}$. The second
one is the alternation between filtering and splitting the
strategies. Both improvements are aimed at reducing the number of uniform
strategies to consider. The improved algorithm is presented in
Algorithm~\ref{algo:eval<>improved}. Using this algorithm, we can
compute $\llbracket \langle \Gamma \rangle \psi\rrbracket^F_{po}$ as
$Improved\llbracket \langle \Gamma \rangle \psi\rrbracket^F_{po}|_{S \times Act_\Gamma}$. The intuition behind
Algorithm~\ref{algo:eval<>improved} is to start by computing the set
of states satisfying the property and the associated actions
(line~\ref{line:eval<>improved:eval}), then get all conflicts
(line~\ref{line:eval<>improved:conflict}) and, if there are conflicts,
choose one conflicting equivalence class of states and possible
actions (lines \ref{line:eval<>improved:pick} to \ref{line:eval<>improved:actions})
 and for each possible action
$a_\Gamma$, recursively call the algorithm with the strategies
following $a_\Gamma$ (lines \ref{line:eval<>improved:strat} and
\ref{line:eval<>improved:call})---i.e. split the class into uniform
strategies for this class and recursively call the algorithm on each
strategy.

\begin{algorithm}[h!t]
	\dontprintsemicolon
	\linesnumberedhidden
	\KwData{$M$ a given (implicit) model, $\Gamma$ a subset of agents of $M$, $\psi$ an $ATLK^F_{po}$ path formula, $Strats \subseteq S \times Act_\Gamma$.}
	\KwResult{The set of states of $M$ satisfying $\langle \Gamma \rangle \psi$ in $Strats$.}
	\BlankLine
	\showln{\label{line:eval<>improved:eval}}$Z = \llbracket \langle \Gamma \rangle \psi\rrbracket^{F,ac}_{fo}|_{Strats}$\;
	\showln{\label{line:eval<>improved:conflict}}$C = \{ \langle s, a_\Gamma \rangle \in Z | \exists \langle s', a'_\Gamma \rangle \in Z \textrm{ such that } s \sim_\Gamma s' \wedge a_\Gamma \neq a'_\Gamma \}$\;
	\If{$C = \emptyset$}{
		\showln{\label{line:eval<>improved:return}}\Return $\{s \in S | \exists a_\Gamma
\in Act_\Gamma \textrm{ s.t. } \forall s'\sim_\Gamma s, \langle
s', a_\Gamma \rangle \in Z\}$\;
	}
	\Else{
		\showln{\label{line:eval<>improved:pick}}$\langle s, a_\Gamma \rangle = $ pick one in $C$\;
		\showln{}$E = \{ \langle s', a'_\Gamma \rangle \in Z | s \sim_\Gamma s'\}$\;
		\showln{\label{line:eval<>improved:actions}}$A = \{ a_\Gamma \in Act_\Gamma | \exists \langle s, a_\Gamma \rangle \in E\}$\;
		$sat = \{\}$\;
		\For{$a_\Gamma \in A$}{
			\showln{\label{line:eval<>improved:strat}}$strat = \{ \langle s', a'_\Gamma \rangle \in E | a'_\Gamma = a_\Gamma \} \cup (Z \backslash E)$\;
			\showln{\label{line:eval<>improved:call}}$sat = sat \cup Improved\llbracket \langle \Gamma \rangle \psi\rrbracket^F_{po}|_{strat}$
		}
		\Return $sat$
	}
	\caption{$Improved\llbracket \langle \Gamma \rangle \psi\rrbracket^F_{po}|_{Strats}$}
	\label{algo:eval<>improved}
\end{algorithm}

More in detail, Algorithm~\ref{algo:eval<>improved} returns the set of
states satisfying the property in $Strats$. So, to get the final result,
we have to take all the states satisfying the property in $S\times Act_\Gamma$.
Algorithm~\ref{algo:eval<>improved} uses the
function $\llbracket . \rrbracket^{F,ac}_{fo}|_{strats}$. This function is a modification
of the $\llbracket . \rrbracket^F_{fo}|_{strats}$ function where
actions are linked to states. More precisely, every sub-call to
$\llbracket . \rrbracket^F_{po}$ or $\overline{Fair_{[\Gamma]}}$ is
enclosed by $StatesActions_\Gamma|_{strats}$ to get all enabled actions in these
states, restricted to $strats$---$StatesActions_\Gamma|_{strats}(Z)$ = $\{\langle s, a_\Gamma \rangle \in  strats
| s \in Z \wedge a_\Gamma \in enabled(s, \Gamma)\}$---, and $Pre_{\langle\Gamma\rangle}|_{strats}$ is replaced by
$Pre^{ac}_{\langle\Gamma\rangle}|_{strats}(Z)$ = $\{ \langle s, a_\Gamma \rangle \in
strats | a_\Gamma \in enabled(s, \Gamma) \wedge \forall a, a_\Gamma
\sqsubseteq a \implies img(s, a) \subseteq Z\}$.
For example, $\llbracket [\Gamma] G \phi
\rrbracket^{F,ac}_{fo}|_{Strats} = \nu Z
. (StatesActions_\Gamma|_{Strats}(\llbracket \phi \rrbracket^F_{po} \cup
\overline{Fair_{[\Gamma]}})) \cap
Pre^{ac}_{\langle\Gamma\rangle}|_{Strats}(Z)$.

Intuitively, $StatesActions_\Gamma|_{strats}(Z)$ returns all the states of
$Z$ with their enabled actions allowed by $strats$ and $Pre^{ac}_{\langle \Gamma
  \rangle}|_{strats}(Z)$ returns the states that can enforce to reach
$Z$ in one step, and the actions that allow them to do so, restricted to actions in $strats$. $\llbracket
\langle \Gamma \rangle \psi\rrbracket^{F,ac}_{fo}|_{strats}$ thus
returns the states satisfying $\langle \Gamma \rangle \psi$ associated
to the actions of $strats$ that allow them to do so.

The correctness of Algorithm~\ref{algo:eval<>improved} is given by the
following theorem.
\begin{theorem}
$Improved\llbracket \langle \Gamma \rangle \psi\rrbracket^F_{po}|_{S\times Act_\Gamma}$ computes the set of states of $M$ satisfying $\langle \Gamma \rangle \psi$, i.e.
\begin{equation*}
\forall s \in S, s \in Improved\llbracket \langle \Gamma \rangle \psi\rrbracket^F_{po}|_{S\times Act_\Gamma} \textrm{ iff } s \models^F_{po} \langle \Gamma \rangle \psi.
\end{equation*}
\end{theorem}
\begin{proof}[Proof sketch]
First, $\llbracket \langle \Gamma \rangle \psi\rrbracket^{F,ac}_{fo}|_{Strats}$ returns the set of states $s$ (and
associated actions) such that there exists a global strategy in
$Strats$ allowing $\Gamma$ to enforce the property in $s$. This means that if a state/action pair is not returned, $\Gamma$ has no global strategy to enforce the property from the given state by using the action given in the pair. By extension, there is no uniform strategy to enforce the property neither. Thus, only state/action pairs returned by $\llbracket \langle \Gamma \rangle \psi\rrbracket^{F,ac}_{fo}|_{Strats}$ have to be considered when searching for a uniform strategy in $Strats$. This also means that $\llbracket \langle \Gamma \rangle \psi\rrbracket^{F,ac}_{fo}|_{Strats}$ filters $Strats$ to winning global strategies; if the result is also a uniform strategy, all the states in the returned set have a uniform strategy to enforce the property.

Second, $Improved\llbracket \langle \Gamma \rangle
\psi\rrbracket^F_{po}|_{Strats}$ returns the set of states satisfying the property in $Strats$.
We can prove this by induction on the number of conflicting
equivalence classes of $Strats$: this is true if there are no conflicting classes because Line~\ref{line:eval<>improved:eval} computes a winning uniform strategy---as discussed above---and Line~\ref{line:eval<>improved:return} returns the set of states for which the strategy is winning for all indistinguishable states. This is also true in the inductive case
because (1) filtering with $\llbracket \langle \Gamma \rangle \psi\rrbracket^{F,ac}_{fo}|_{Strats}$ doesn't lose potential state/action pairs and (2) the algorithm takes one conflicting class and tries all the
possibilities for this class.

The final result thus is correct since it returns the set of
states $s$ for which there is a uniform strategy in $S\times Act_\Gamma$ that is winning for all
states equivalent to $s$.
\end{proof}

\paragraph{Complexity considerations}
Model checking $ATL$ with perfect recall and partial observability is
an undecidable problem~\cite{Schobbens-04}, while model checking
$ATL_{ir}$ is a $\Delta_2^P$-complete problem~\cite{Jamroga-Dix-06}. $ATLK^F_{po}$ subsumes
$ATL_{ir}$ and its model checking problem is therefore
$\Delta_2^P$-hard. Algorithm~\ref{algo:eval<>} performs a call
to $[[.]]^F_{fo}$ for each uniform strategy: $[[.]]^F_{fo}$ is in {\bf
  P}, but in the worst case there could be exponentially many calls to
this procedure, as there could be up to $\prod_{i \in \Gamma}
|Act_i|^{|S_i|}$ uniform strategies to consider.

\section{Conclusion}\label{section:conclusion}

A number of studies in the past have investigated the problem of model
checking strategies under partial observability and, separately, some
work has provided algorithms for including fairness constraints on
\emph{actions} in the case of full observability. To the best of our
knowledge, the issue of fairness constraints and partial observability
have never been addressed together. 

In this paper we presented $ATLK^F_{po}$, a logic combining partial
observability and fairness constraints on \emph{states} (which is the
standard approach for temporal and epistemic logics), and we have
provided a model checking algorithm.
The proposed algorithm is similar to the one of Calta et al.~\cite{Calta-Shkatov-others-10}. They also split possible actions into uniform strategies, but they do not provide a way to deal with fairness constraints.

Finally, the structure of our algorithm is compatible with symbolic model
checking using OBDDs, and we are working on its implementation in the
model checker MCMAS~\cite{LomuscioQuRaimondi09}, where fairness
constraints are only supported for temporal and epistemic operators.

\bibliographystyle{eptcs}
\bibliography{ref}

\end{document}